\documentclass[10pt, final, twocolumn, twoside, romanappendices]{IEEEtran}
\usepackage{xcolor}
\usepackage{mathtools}
\usepackage{epsfig,makeidx,color}
\usepackage{amsmath,amssymb,bbm,enumitem}
\usepackage{mathrsfs} 
\usepackage{cite,graphicx,lipsum}
\usepackage[switch,pagewise]{lineno}
\usepackage{hyperref}
\usepackage[caption=false,font=footnotesize]{subfig}
\usepackage{algorithm}
\usepackage{algpseudocode}
\usepackage{acronym}
\usepackage{array}    
\usepackage{booktabs}
\hypersetup{
        colorlinks = true,
        citecolor=red,
}
\usepackage{tikz}
\usepackage{amsthm}

\newtheorem{proposition}{Proposition}

\def\be{ \begin{equation} }
\def\ee{ \end{equation} }
\def\bea{ \begin{eqnarray} }
\def\eea{ \end{eqnarray} }

\def\b0{{\bf 0}}

\acrodef{HO}{handover}
\acrodef{TX}{transmitter}
\acrodef{RX}{receiver}
\acrodef{cHO}{cache-based \ac{HO}}
\acrodef{tHO}{token-based \ac{HO}}
\acrodef{ctHO}{joint cache-token based \ac{HO} framework}
\acrodef{BS}{base station}
\acrodef{UE}{user equipment}
\acrodef{LLM}{large language model}
\acrodef{SNR}{signal-to-noise ratio}

\ifCLASSOPTIONonecolumn
\else

\fi

\usepackage{geometry}
\graphicspath{ {./images/} }
\usepackage{graphicx}
\begin{document}
\title{Low-Latency Edge LLM Handover via \\Joint KV Cache Transfer and Token Prefill}



\author{
   \vspace{0.0cm}
    Seunghun~Lee,~\IEEEmembership{Graduate Student Member,~IEEE},
    Jihong~Park,~\IEEEmembership{Senior~Member,~IEEE},\\
    Ce Zheng,~\IEEEmembership{Member,~IEEE},
    and Hyuncheol~Park,~\IEEEmembership{Senior~Member,~IEEE}

    \thanks{
        S. Lee and H. Park are with
        the School of Electrical Engineering,
        Korea Advanced Institute of Science and Technology,
        Daejeon 34141,
        Republic of Korea
        (e-mail: {seunghun21@kaist.ac.kr; hcpark@kaist.ac.kr}).

        J. Park is with
        ISTD Pillar, Singapore University of Technology and Design (SUTD),
        8 Somapah Rd, Singapore 487372, 
        Singapore
        (e-mail: {jihong\_park@sutd.edu.sg}).

        C. Zheng is with the Department of Broadband Communication, 
        Pengcheng Lab,
        Shenzhen 518066, China
        (e-mail: {zhengc@pcl.ac.cn}).
    }
    \thanks{
This work was supported in part by A$^\ast$STAR under its IAF-ICP (I2501E0064), in part by the IITP-ITRC grant funded by the Korean government (MSIT) (IITP-2026-RS-2023-00259991) (33\%), in part by SUTD Kickstarter Initiative (SKI 2021\_06\_08), and in part by the National Research Foundation, Singapore, and the Infocomm Media Development Authority under its Future Communications Research \& Development Programme.
     (\textit{Corresponding authors: 
     J. Park and H. Park})
     }

\vspace{-0.8cm}
}

\maketitle
\begin{abstract} 
Edge deployment of large language models (LLMs) can reduce latency for interactive services, but mobility introduces service interruptions when an user equipment (UE) hands over between base stations (BSs). To promptly resume decoding, the target-side edge server must recover the UE context state, which can be provisioned either by token forwarding followed by prefill computation or by direct key-value (KV) cache transmission over backhaul. This paper proposes a unified handover (HO) design that jointly selects the prefill length and schedules backhaul KV cache delivery to minimize the worst-user LLM HO delay for multiple UEs. The resulting scheme admits a tractable step-wise solution with explicit feasibility conditions and a constructive rate-scheduling policy. Simulations show that the proposed method consistently outperforms baselines across a wide range of backhaul capacities, prefill speeds, and context sizes, providing practical guidelines for mobility-aware Edge LLM token streaming.

\end{abstract}


\ifCLASSOPTIONonecolumn
\baselineskip 28pt
\fi
\section{Introduction}

\IEEEPARstart{L}{arge} language model (LLM) streaming has recently emerged as a key application for beyond 5G networks, with billions of users already accessing services such as ChatGPT and Gemini via mobile devices. Existing LLM streaming services are predominantly cloud-based, which can incur significant and time-varying latency over wireless links \cite{Yet2026}. To enable low-latency and differentiated services, deploying LLMs at the network edge, referred to as \emph{Edge LLM}, has recently been recognized as a promising approach \cite{Jang25,WG3}.

However, provisioning seamless Edge LLM streaming for mobile users remains challenging due to token decoding dependency during \emph{Edge LLM \ac{HO}}. Specifically, LLM inference is autoregressive, where each token is generated conditioned on the key-value (KV) cache of previously decoded tokens. When a mobile user is handed over between \acp{BS} equipped with Edge LLMs, the target BS has not decoded the user’s past tokens, and therefore cannot immediately preserve the token streaming context.

A straightforward solution is to transfer past tokens to the target BS and re-decode them to reconstruct the KV cache. The KV cache reconstruction of this \emph{\ac{tHO}} is equivalent to the prefill phase of LLM inference, which is computationally intensive and incurs a large time-to-first-token (TTFT), resulting in substantial HO delay. Furthermore, prefill is often batched across users, making the worst-user delay the bottleneck under simultaneous HOs.
Alternatively, BSs can transfer the KV cache directly over data backhaul links. This \emph{\ac{cHO}} can significantly reduce Edge LLM HO delay without re-decoding past tokens. However, the KV cache can be large (e.g., GB-scale for billion-parameter models), and limited backhaul capacity may constrain inter-BS KV cache transfer, particularly when multiple HOs occur simultaneously.

Motivated by these challenges, we propose a \emph{\ac{ctHO}} for Edge LLM systems, in which token-based partial KV cache reconstruction is performed concurrently with the remaining KV cache transfer as shown in Fig.~\ref{fig:schematic_illustration}. This joint design minimizes the worst-user HO delay under backhaul capacity constraints.

\begin{figure}[t]
\centering
\includegraphics[width=\columnwidth]{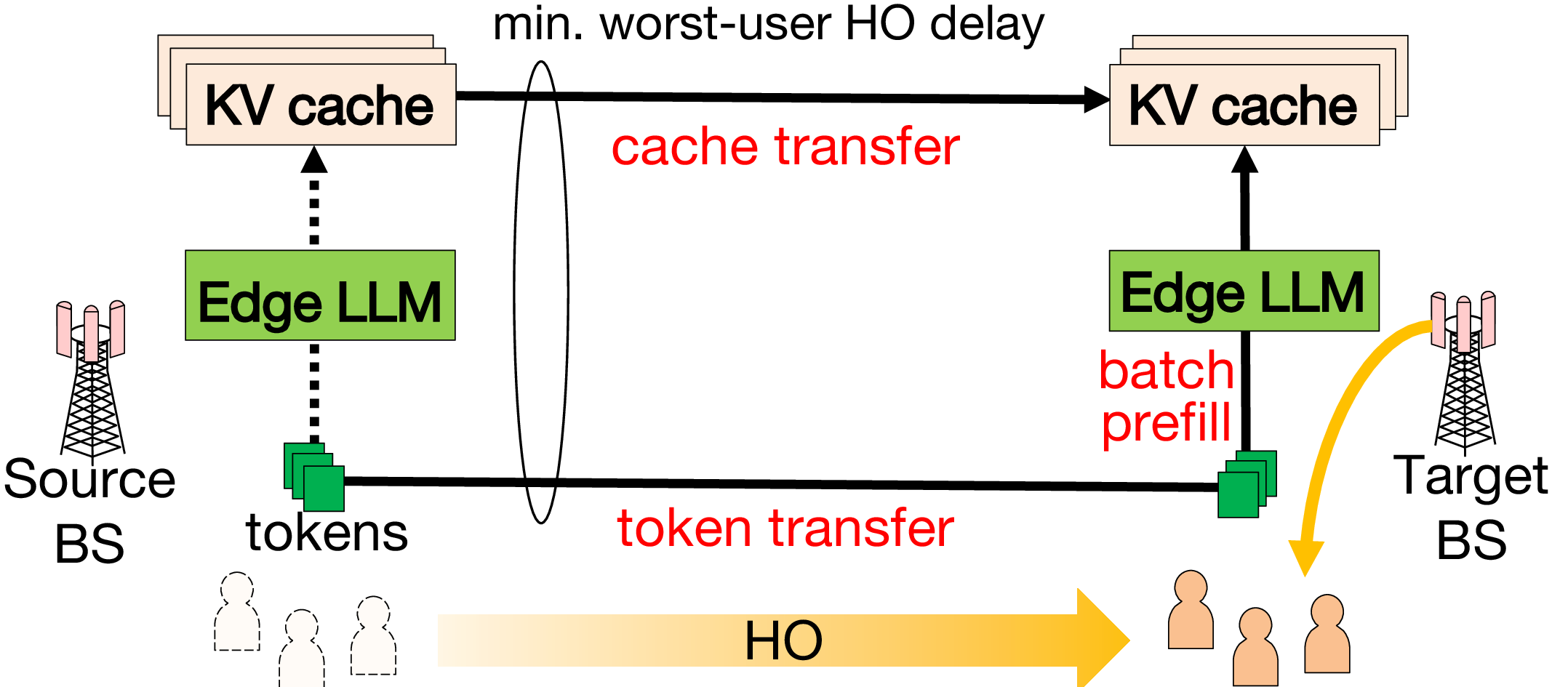}
\vspace{-0.8cm}
\caption{An illustration of the proposed \ac{ctHO}, which jointly exploits batch prefill, and KV cache transfer to minimize the worst-user \ac{HO} delay.}
\label{fig:schematic_illustration}
\end{figure}




The main contributions of this paper are as follows:
\begin{itemize}
\item We propose ctHO that minimizes Edge LLM \ac{HO} delay by optimizing the length of the prefilled KV cache and the backhaul rate allocation.

\item We develop a two-stage optimization that determines the backhaul scheduling for a given prefill length and then selects the prefill length, with a proof of global optimality.

\item Simulations show that ctHO reduces the worst-user \ac{HO} delay by up to $3.1\times$ and $3.7\times$ compared to \ac{cHO} and \ac{tHO}, respectively.
\end{itemize}

Token-based communication has been studied for multiple access \cite{Mahdi}, packetization \cite{Lee2025}, and multimodal transmission \cite{Qiao2025b}, but existing works assume static users. In contrast, we focus on \ac{HO}s of mobile users. Meanwhile, KV cache management, including compression \cite{liu2024} and cache-based inter-agent communication \cite{C2C, Chen2026}, has recently emerged. These techniques are orthogonal to our approach and can further reduce the Edge LLM \ac{HO} latency of the proposed framework.


\section{System Model and Problem Formulation}\label{sec:system_model}
\subsection{Multi-UE HO During Edge LLM Token Streaming}
We consider a mobile Edge LLM token streaming service in which a source \ac{BS} separately streams LLM-generated tokens to \(K\) UEs. During this service, we assume that the UEs move toward a target \ac{BS}, and that UE \(i\) triggers \ac{HO} to the target \ac{BS} at time \(\tau_i\), where both BSs deploy the same LLM architecture and parameters.
After \ac{HO} occurs, token streaming is stopped and the target \ac{BS} prepares the KV cache corresponding to the \(C_i\) tokens, where \(C_i\) denotes the number of tokens decoded for UE \(i\) up to \(\tau_i\). Streaming resumes only after the KV cache preparation is completed, and the resulting Edge LLM \ac{HO} delay is the objective to minimize.

As shown in Fig.~\ref{fig:schematic_illustration}, the target BS can rebuild a portion of the required KV cache by receiving tokens from the source \ac{BS} and performing prefill on them, while the remaining KV cache can be directly delivered via backhaul \cite{liu2024}. The delay for token transmission is assumed to be negligible compared to the KV cache transfer and prefill delays.
Because token streaming can resume only after both processes are completed, the LLM \ac{HO} delay is governed by the larger of the two delays. We therefore aim to minimize the LLM \ac{HO} delay by jointly determining the amount of KV cache processed by prefill and the backhaul capacity allocation across UEs. To this end, we assume that the source BS knows \(\{(\tau_i, C_i)\}_{i=1}^{K}\) and the prefill delay profile of the target BS.
Without loss of generality, we sort UEs by their \ac{HO} trigger times,
\begin{align}
\tau_1 \le \tau_2 \le \cdots \le \tau_K.
\end{align}

At the target BS, to reduce the prefill delay compared with sequentially performing separate prefills for individual UEs, batch prefill jointly processes the tokens of multiple UEs in a single batch prefill \cite{Zheng2024}.
Accordingly, the KV cache is reconstructed for a common prefix of length \(L\) for all UEs, where \(L\in[0,C_{\max}]\) and \(C_{\max}\triangleq \max_i C_i\).
Because the same input token sequence length must be used across the batch, if \(C_i<L\), only the available \(C_i\) tokens of UE \(i\) are involved in batch prefill, while the remaining token positions up to length \(L\) are zero-padded. Hence, the \(C_i\) tokens of UE \(i\) are split as
\begin{align}
n_i^{\textrm{(pf)}}(L) &\triangleq \min\{C_i,\,L\}, \\
n_i^{\textrm{(tx)}}(L) &\triangleq C_i - n_i^{\textrm{(pf)}}(L)
= \bigl(C_i-L\bigr)^{+},
\end{align}
where $(\cdot)^+=\max\{\cdot,0\}$, \(n_i^{\textrm{(pf)}}(L)\) is the number of tokens processed by prefill, and \(n_i^{\textrm{(tx)}}(L)\) is the number of remaining tokens whose corresponding KV cache is delivered from the source BS.

For backhaul transmission, the rates allocated to the UEs at time \(t\), denoted by \(r_i(t)\), are subject to the capacity constraint
\begin{align}
\sum_{i=1}^{K} r_i(t) \le R \;\;[\text{tokens/s}],\quad \forall t,
\end{align}
where \(R \triangleq R_{\mathrm{bh}}/s_{\mathrm{KV}}\) is the backhaul capacity normalized by the KV cache payload size per token \(s_{\mathrm{KV}}\). Thus, both \(R\) and \(r_i(t)\) are expressed in tokens/s to match the token-based quantities in the prefill model, such as \(L\) and \(C_i\).

\subsection{Worst-user LLM HO Delay Minimization}\label{sec:model_problem}
Let \(T^{\mathrm{(pf)}}(L)\) and \(T_i^{\mathrm{(tx)}}(L,r_i,C_i)\) denote the completion times of batch prefill and cache transfer for UE \(i\), respectively. Based on these completion times, the \ac{HO} prefill delay and the cache transfer delay are defined as
\begin{align}
D_i^{\mathrm{(pf)}}(L) &\triangleq T^{\mathrm{(pf)}}(L)-\tau_i,\\
D_i^{\mathrm{(tx)}}(L,r_i,C_i) &\triangleq T_i^{\mathrm{(tx)}}(L,r_i,C_i)-\tau_i,\label{eq:D_tx}
\end{align}
respectively. Since token decoding can resume only after both processes are completed and KV cache is reconstructed, the resulting LLM \ac{HO} delay of UE \(i\) is
\begin{align}
D_i(L,r_i,C_i) \triangleq \max\left\{D_i^{\mathrm{(pf)}}(L),\,D_i^{\mathrm{(tx)}}(L,r_i,C_i)\right\}.
\end{align}
Accordingly, the worst-user LLM \ac{HO} delay is defined as
\begin{align}
D(L,r) \triangleq \max_{i\in\{1,\dots,K\}} D_i(L,r_i,C_i).
\end{align}

We next specify the completion times of batch prefill and cache transfer. Let \(T_c\) denote the batch prefill cycle interval. Then batch prefill starts at the first cycle boundary after the latest \ac{HO} time \(\tau_{K}\), i.e., $t_s \triangleq \left\lceil \frac{\tau_{K}}{T_c}\right\rceil T_c.$
Hence, the batch prefill completion time is
\begin{align}
T^{\mathrm{(pf)}}(L) \triangleq t_s + p(L),
\end{align}
where \(p(L)\) denotes the batch prefill delay. The cache transfer completion time $T_i^{\mathrm{(tx)}}(L,r,C_i)$ is time at which the remaining KV cache corresponding to \(n_i^{\mathrm{(tx)}}(L)\) tokens has been fully delivered.
Accordingly,
\begin{align}
\hspace{-0.2pc}T_i^{\mathrm{(tx)}}(L,r_i,C_i)
\triangleq \inf\left\{t\ge \tau_i\!:\!\int_{\tau_i}^{t}\! r_i(u)\,du\!\ge\!n_i^{\mathrm{(tx)}}(L)\!\right\}.
\end{align}

Our goal is to jointly choose the shared prefill length and the backhaul rate allocation to minimize the worst-user LLM \ac{HO} delay over all UEs:
\refstepcounter{equation}\label{eq:P}
\begin{align}
{\mathscr P}:&\min_{0\le L\le C_{\max},\; r(\cdot)} \; D(L,r)\tag{\theequation a} 
\\&\hspace{1.7pc} {~{\text {s.t.}}}\hspace{0.3pc}\sum_i r_i(t)\le R.\hspace{-0.5pc}\tag{\theequation b}&& \label{eq:backhaul_constraint}
\end{align}
Although \({\mathscr P}\) is a joint optimization problem over \(L\) and \(r(\cdot)\), it can be solved in a step-wise manner without loss of optimality. Specifically, for each fixed \(L\), we first optimize \(r(\cdot)\), and then optimize \(L\). The following proposition shows that this decomposition is globally optimal.
To formalize this step-wise decomposition, define the value function
\begin{align}
V(L) \triangleq \min_{r(\cdot)} D(L,r).
\end{align}
Let $r^{\star}(L)\in\arg\min_{r(\cdot)} D(L,r)$ denote an optimal rate allocation policy for a given $L$ and the backhaul constraint so that $V(L)=D(L,r^*(L))$, and let $L^{\star} = \arg\min_{0\le L\le C_{\max}} V(L)$.
\begin{proposition}
\label{prop:global}
The pair $\big(L^{\star},\, r^{\star}(\cdot)\big)$ is a global optimum of~${\mathscr P}$.
\end{proposition}
\begin{proof}
For any feasible $(L,r)$, by definition of $V(L)$ we have $V(L)\le D(L,r)$. Moreover, $D\big(L^{\star},r^{\star}(L^{\star})\big)=V(L^{\star})$ and $V(L^{\star})\le V(L)$ by optimality of $L^{\star}$. Hence,
\begin{equation}
D\big(L^{\star},r^{\star}(L^{\star})\big)=V(L^{\star}) \le V(L)\le D(L,r),
\end{equation}
which proves global optimality.
\end{proof}

\section{Step-wise Optimization of Backhaul Scheduling and Prefill Length}
\label{sec:solution}


Based on Proposition~\ref{prop:global}, we decompose ${\mathscr P}$ into the following two subproblems. \eqref{eq:P} is equivalently expressed as
\begin{align}
{\mathscr P}_1:\quad &V(L) \coloneq\min_{r(\cdot)} D(L,r) ,\quad \text{for a given }L,
\label{eq:P1}\\
{\mathscr P}_2:\quad &\min_{0\le L\le C_{\max}} V(L).
\label{eq:P2}
\end{align}
In the following, we first characterize $V(L)$ by solving ${\mathscr P}_1$ for any fixed $L$, then obtain the optimal $L$ by solving ${\mathscr P}_2$.

\subsection{Optimal Cache Transfer Delay and Prefill Length}
\label{sec:closedform}

To solve ${\mathscr P}_1$ for a fixed $L$, since prefill delay is independent of $r(\cdot)$, optimizing ${\mathscr P}_1$ reduces to minimizing the worst-user cache tramsfer delay $D^{(\mathrm{tx})}(L,r)$ as follows.
\begin{align}\label{eq:P1_split}
V(L)
&= \min_{r(\cdot)} D(L,r) \notag\\
&= \min_{r(\cdot)} \max_i \max\Bigl\{D_i^{\mathrm{(pf)}}(L),\, D_i^{\mathrm{(tx)}}(L,r_i,C_i)\Bigr\} \notag\\
&= \min_{r(\cdot)} \max\Bigl\{\max_i D_i^{\mathrm{(pf)}}(L),\, \max_i D_i^{\mathrm{(tx)}}(L,r_i,C_i)\Bigr\} \notag\\
&\triangleq \max\Bigl\{D^{\mathrm{(pf)}}(L),\, \min_{r(\cdot)} D^{\mathrm{(tx)}}(L,r)\Bigr\},
\end{align}
where $D^{\mathrm{(pf)}}\!(L)\!\triangleq\!\max_i\!D_i^{\mathrm{(pf)}}\!(L)$ and $ D^{\mathrm{(tx)}}(L,r)\!\triangleq\!\max_{i}\!D_i^{\mathrm{(tx)}}(L,r_i,C_i).$

The following proposition characterizes the minimum achievable cache transfer delay for a fixed \(L\) by examining the cumulative remaining KV caches over the first \(k\) UEs. To this end, define
\begin{align}
S_k(L)\triangleq \sum_{i=1}^{k} n_i^{\textrm{(tx)}}(L),\quad k\in\{1,\dots,K\},
\label{eq:Sk_def}
\end{align}
as the cumulative remaining token amount for the first \(k\) UEs.


\begin{figure}[t]
\centering
\includegraphics[width=0.95\columnwidth]{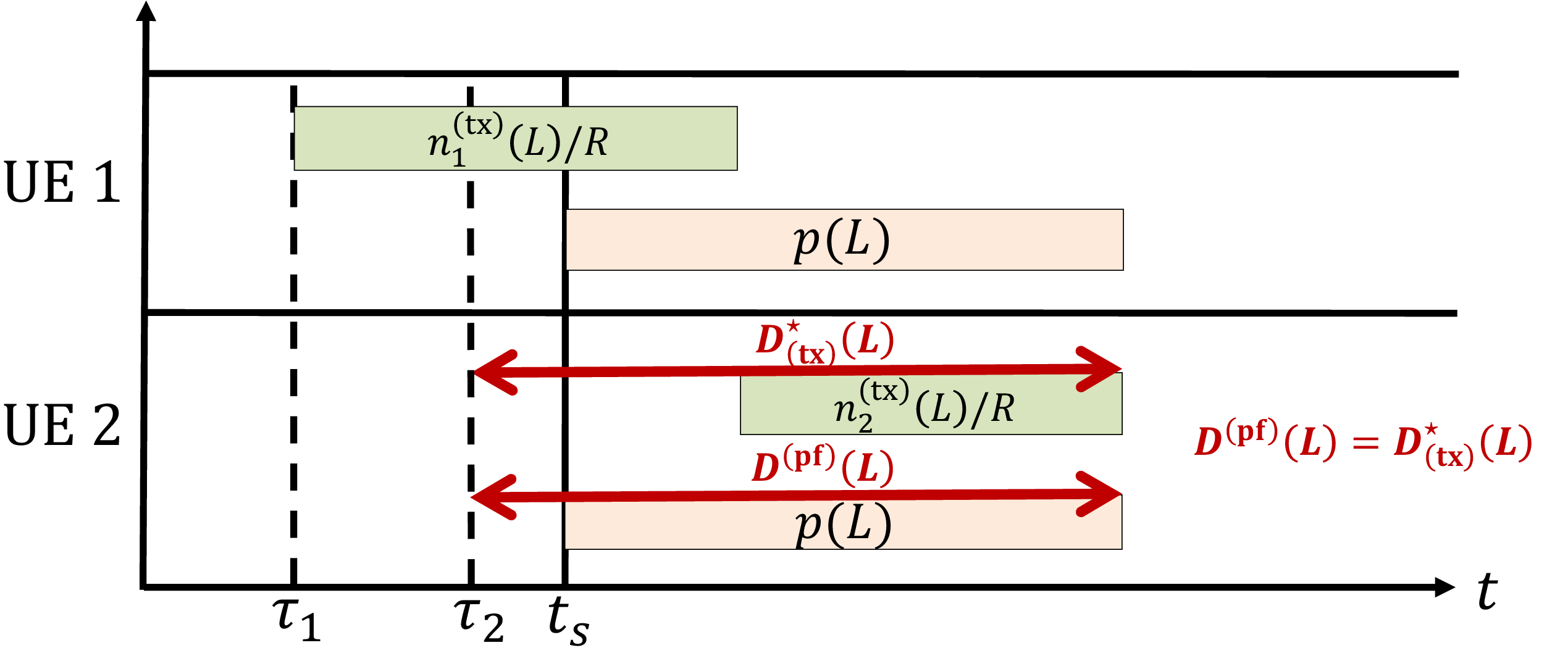}
\vspace{-0.3cm}
\caption{An illustration of the proposed design principle. The \ac{HO} prefill delay \(D^{(\mathrm{pf})}(L)\) and the minimum cache transfer delay $D_{(\mathrm{tx})}^{\star}(L)$ are set to equal by choosing the batch prefill length \(L\), so that neither process dominates the overall LLM \ac{HO} delay.}
\label{fig:ctHO_principle}
\end{figure}

\begin{proposition}
\label{prop:Dtx_closed}
For any fixed $L\in[0,C_{\max}]$, the minimum feasible cache transfer delay $D_{(\mathrm{tx})}^{\star}(L) \!\triangleq\!\min_{r(\cdot)}\!D^{(\mathrm{tx})}(L,r)\!$ is
\begin{align}
D_{(\mathrm{tx})}^{\star}(L)
=
\max_{k\in\{1,\dots,K\}}
\left[
\frac{S_k(L)}{R} - \big(\tau_k-\tau_1\big)
\right]^{+}.
\label{eq:Dtx_closed}
\end{align}
\end{proposition}
\begin{proof}
For any feasible rate allocation $r(\cdot)$, by the definition $D^{(\mathrm{tx})}(L,r) = \max_i D_i^{(\mathrm{tx})}(L,r_i,C_i)$ and \eqref{eq:D_tx}, every UE satisfies
\begin{align}\label{eq:tx_deadline_feas}
T_i^{(\mathrm{tx})}(L,r_i,C_i) \le \tau_i + D^{(\mathrm{tx})}(L,r), \quad \forall i \in \{1,\dots,K\}.
\end{align}
Since the KV cache of UE $i$ can be transmitted only after $\tau_i$, by time $\tau_k + D^{(\mathrm{tx})}(L,r)$ the system must have delivered at least the total KV cache amount corresponding to $S_k(L)$ tokens released up to $\tau_k$ to the first $k$ UEs.
On the other hand, under the constraint $\sum_i r_i(t) \le R$, the source BS can deliver the KV cache corresponding to at most $R\big(D^{(\mathrm{tx})}(L,r) + \tau_k - \tau_1\big)$ tokens over $[\tau_1,\, \tau_k + D^{(\mathrm{tx})}(L,r)]$.
Therefore,
\begin{align}
S_k(L) \le R\Big(D^{(\mathrm{tx})}(L,r) + \tau_k - \tau_1\Big), \,\, \forall k \in \{1,\dots,K\}.
\label{eq:feas_ineq}
\end{align}
Rearranging and combining with $D^{(\mathrm{tx})}(L,r) \ge 0$ yields
\begin{align}
D^{(\mathrm{tx})}(L,r) \ge \left[\frac{S_k(L)}{R} - \big(\tau_k - \tau_1\big)\right]^{+}, \quad \forall k,
\end{align}
which holds for every feasible $r(\cdot)$, so $D_{(\mathrm{tx})}^{\star}(L) \ge \max_k \left[\frac{S_k(L)}{R} - (\tau_k - \tau_1)\right]^{+}$. Achievability follows from the policy in Section~\ref{sec:constructive} that attains this bound.
\end{proof}

In the following, we consider only the UEs whose remaining KV cache cannot be fully delivered before batch prefill, since batch prefill is unnecessary for UEs whose remaining KV cache can already be fully delivered.
We next solve ${\mathscr P}_2$ by optimizing $L$ in \eqref{eq:P2}. Since $D^{(\mathrm{pf})}(L)$ is non-decreasing in $L$ whereas $D_{(\mathrm{tx})}^{\star}(L)$ is non-increasing in $L$, the minimum is achieved by choosing \(L\) such that the two terms are equal, whenever possible. The following proposition formalizes this condition.


\begin{proposition}
\label{prop:L_equalize}
If there exists an $L\in[0,C_{\max}]$ such that
\begin{align}
D^{(\mathrm{pf})}(L)= D_{(\mathrm{tx})}^{\star}(L),
\label{eq:equalize}
\end{align}
then any such $L$ minimizes $V(L)$ on $[0,C_{\max}]$. Otherwise, the minimizer occurs at a boundary, i.e., $L^{\star}\in\{0,C_{\max}\}$.
\end{proposition}

\begin{proof}
Let $f(L)\triangleq D^{(\mathrm{pf})}(L)$ and $g(L)\triangleq D_{(\mathrm{tx})}^{\star}(L)$, where $f$ is non-decreasing and $g$ is non-increasing on $[0,C_{\max}]$, and And according to \eqref{eq:P1_split}, $V(L)=\max\{f(L),g(L)\}$. 
If an $L_0$ satisfies $f(L_0)=g(L_0)$, then $V(L)=g(L)\ge g(L_0)$ for $L<L_0$ and $V(L)=f(L)\ge f(L_0)$ for $L>L_0$, so $L_0$ minimizes $V(L)$. 
Otherwise, either $f(L)\!>\!g(L)$ or $f(L)\!<\!g(L)$ for all $L$, and the minimum is attained at $L=0$ or $L=C_{\max}$, respectively.
\end{proof}

\subsection{Optimal Backhaul Rate Scheduling Algorithm}
\label{sec:constructive}
Proposition~\ref{prop:Dtx_closed} characterizes the optimal value \(D_{(\mathrm{tx})}^{\star}(L)\), whereas \({\mathscr P}_1\) also requires a feasible policy \(r(\cdot)\) that attains it. Among possibly many such optimal policies, we present a simple one that achieves \(D_{(\mathrm{tx})}^{\star}(L)\) by allocating the backhaul rate to one UE at a time.


Let \(\pi(t)\in\{1,\dots,K\}\) denote the backhaul scheduler, i.e., the index of the UE that is served by the backhaul at time \(t\). The scheduler is updated only when the active set changes, namely, when a new UE begins to be served at \(t=\tau_i\) or when the currently served UE completes its remaining KV cache transmission. Let $n_i^{\mathrm{(rem)}}(t,L) \triangleq
\left[n_i^{\mathrm{(tx)}}(L)-\int_{\tau_i}^{t} r_i(s)\,ds\right]^+$ denote the remaining token amount of UE \(i\) at time \(t\). Also, let
\begin{align}
\mathcal{A}(t)\triangleq \left\{
i\in\{1,\dots,K\}: t\ge \tau_i,\; n_i^{\mathrm{(rem)}}(t,L)>0
\right\}
\end{align}
denote the active set of UEs whose \ac{HO} has already occurred but whose KV cache transfer has not yet been completed. For a given \(L\), we define the target cache transfer completion time of UE \(i\) as
$d_i(L)\triangleq \tau_i + D_{(\mathrm{tx})}^{\star}(L)$.
At any time $t$, the scheduler $\pi(t)$ selects the UE $i$ with the smallest $d_i(L)$ among $\mathcal{A}(t)$ and allocates the entire backhaul capacity to that UE, i.e.,
\begin{align}
\pi(t)&=\arg\min_{i\in\mathcal{A}(t)} d_i(L),\\
r_i(t)&=R\cdot \mathbf{1}\{\pi(t)=i\}.
\end{align}
The full backhaul capacity is always assigned to a single UE, and $\pi(t)$ changes only when $\mathcal{A}(t)$ changes (either at $t=\tau_i$ or when $n_{\pi(t)}^{\textrm{(rem)}}(t,L)=0$).

Under the policy above, the backhaul is never idle whenever $\mathcal{A}(t)\neq\emptyset$, and it always serves the UE with the smallest $d_i(L)$ among $\mathcal{A}(t)$.
Hence, by time $\tau_k+D_{(\mathrm{tx})}^{\star}(L)$ the policy can transmit at least $R\big(D_{(\mathrm{tx})}^{\star}(L)+\tau_k-\tau_1\big)$ that is available up to each $\tau_k$, which is sufficient to complete $S_k(L)$ for every $k$. 
This construction yields a feasible $r(\cdot)$ that attains $D_{(\mathrm{tx})}^{\star}(L)$ for the fixed $L$, completing the solution of ${\mathscr P}_1$.

\section{Simulation Results}\label{sec:simulation_results}

\subsection{Simulation Setup}\label{sec:sim_setup}
We consider a multi-UE \ac{HO} scenario with $K=4$ UEs. For each UE \(i\), $C_i$ is uniformly sampled from $\mathrm{Unif}[1024,C_{\max}]$ tokens. $T_{\mathrm{c}}$ is set as $=0.01$~s. A 1D line-network model is adopted with two BSs located at $x^{(\mathrm{s})}=0$ (source) and $x^{(\mathrm{t})}=D_{\mathrm{bs}}$ (target), where $D_{\mathrm{bs}}=300$\,m\cite{Neetu2023}. Each UE starts from \(x_i(0)\sim \mathrm{Unif}[120,130]\)\,m and moves with a constant speed \(v_i=20\)\,m/s, so that its position evolves as \(x_i(t)=x_i(0)+v_i t\). \ac{HO} is triggered when the UE crosses the boundary \(x_{\mathrm{b}}=150\)\,m, which gives the \ac{HO} time $\tau_i=\frac{x_{\mathrm{b}}-x_i(0)}{v_i}.$ This 1D line-network and mobility setting abstracts a road-segment scenario where vehicles move along a line and sequentially associate with nearby access points deployed along the roadside.

The KV cache payload size per token is 
$s_{\mathrm{KV}}
= 2\cdot N_{\ell}\cdot N_{\mathrm{kv}}\cdot d_h \cdot q\;\text{bits}$,
where the leading factor $2$ accounts for the key and value tensors, $N_{\ell}$ is the number of Transformer layers, $N_{\mathrm{kv}}$ is the number of key-value heads, $d_h$ is the head dimension, and $q$ is the per-element precision.
For LLM, we adopt \texttt{Qwen2.5-7B-Instruct} in which $(N_{\ell},N_{\mathrm{kv}},d_h)=(28,4,128)$ and $q=16$ bits.
This yields \(s_{\mathrm{KV}}=458{,}752\) bits per token, so that when \(C=3072\) tokens, the total KV cache size is \(176\) MB.
We use prefill model \(p(L)=aL+b\), where $a,b$ are set as $9.4267\times 10^{-5}$ and $2.4\times10^{-3}$ by default.


\begin{figure*}[t]
    \centering
    \subfloat[Worst-user LLM HO delay vs. backhaul rate \ensuremath{R_{\mathrm{bh}}} (Gbps).]{
        \includegraphics[width=0.23\textwidth]{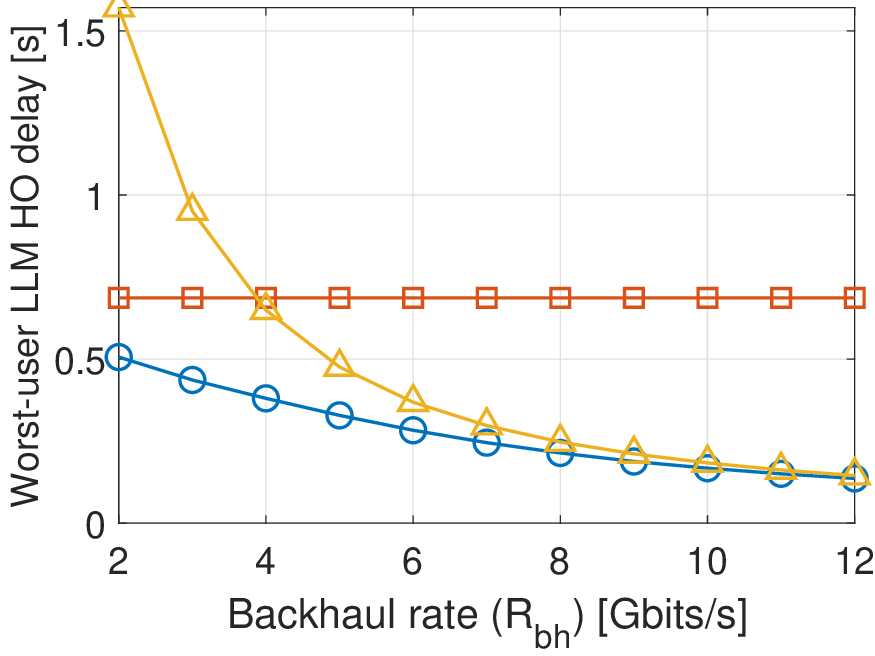}
        \label{fig:plot_rate_sweep_v1}
    }\hfill
    \subfloat[Worst-user LLM HO delay vs. prefill speed \ensuremath{1/a} with \ensuremath{R_{\mathrm{bh}}=4.5}\,Gbps.]{
        \includegraphics[width=0.23\textwidth]{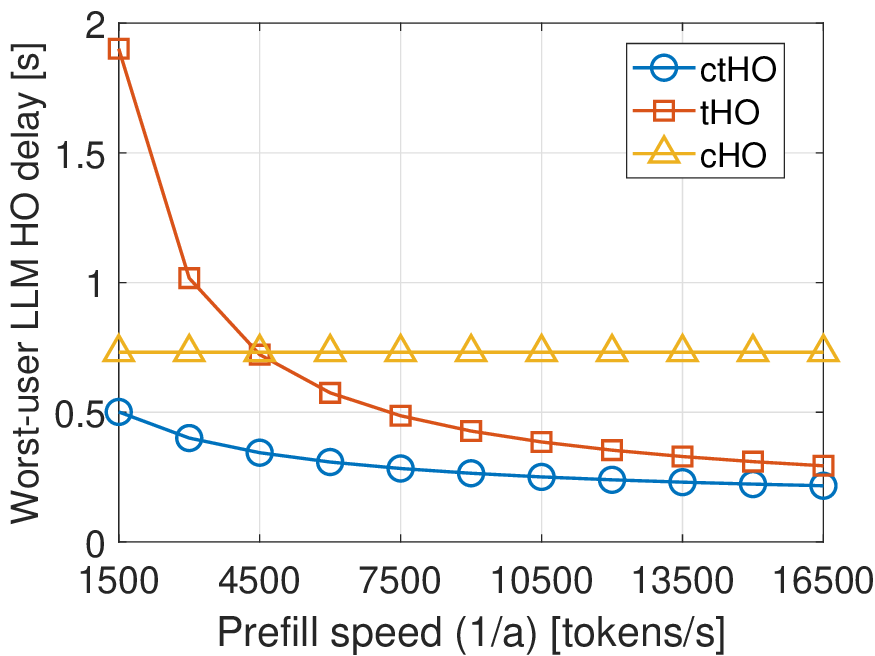}
        \label{fig:plot_compute_sweep_v1}
    }\hfill
    \subfloat[Worst-user LLM HO delay vs. the maximum cache size \ensuremath{C_{\max}} with \ensuremath{R_{\mathrm{bh}}=4.5}\,Gbps.]{
        \includegraphics[width=0.23\textwidth]{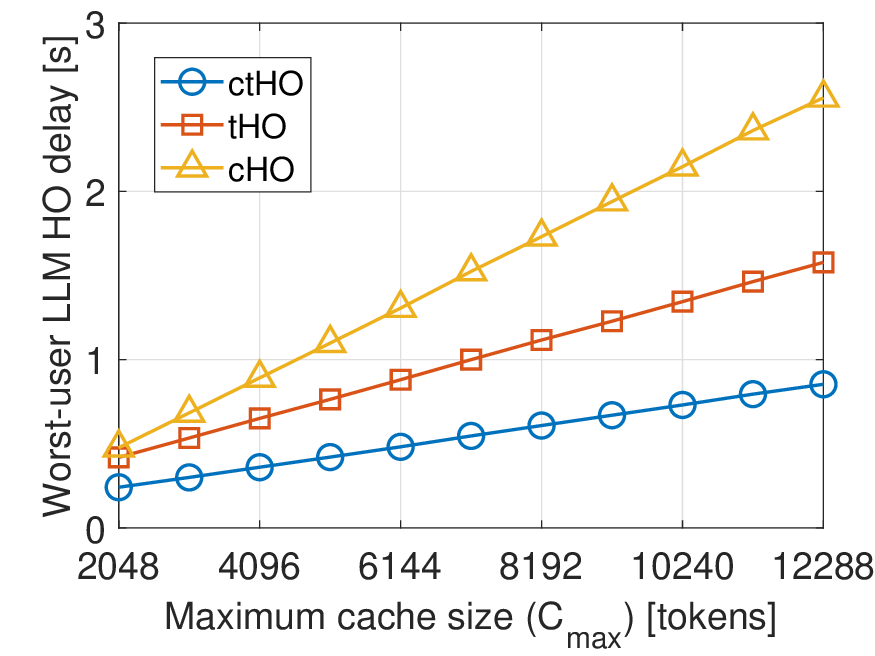}
        \label{fig:plot_cache_sweep_v1}
    }\hfill
    \subfloat[Worst-user LLM HO delay vs. the number of UEs \ensuremath{K} with \ensuremath{R_{\mathrm{bh}}=4.5}\,Gbps.]{
        \includegraphics[width=0.23\textwidth]{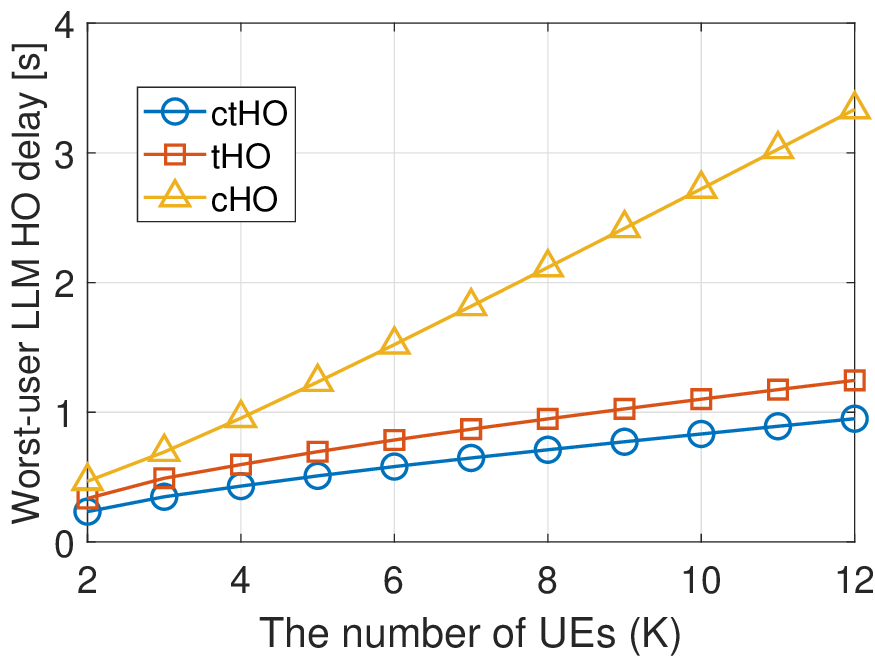}
        \label{fig:plot_K_sweep_v1}
    }
    \caption{Worst-user LLM HO delay comparison of \ac{ctHO}, tHO, and cHO under different system parameters.}
    \label{fig:sweep_quad}
\end{figure*}


\subsection{Impact of Backhaul Transmission Rate, Prefill Speed, Cache Size, and Number of UEs}\label{sec:Rsweep}
In this section, we compare \ac{ctHO} with \ac{tHO} and \ac{cHO}.
\begin{align}
&\textit{tHO:}\ L=C_{\max} \Rightarrow\ n_i^{\textrm{(pf)}}(L)\!=\!C_i,\ n_i^{\textrm{(tx)}}(L)\!=\!0, \forall i,\\
&\textit{cHO:}\ L=0 \Rightarrow\ n_i^{\textrm{(pf)}}(L)\!=\!0,\ n_i^{\textrm{(tx)}}(L)\!=\!C_i, \forall i.
\end{align}
where \ac{tHO} relies entirely on batch prefill, whereas \ac{cHO} relies entirely on KV cache transmission over backhaul.
Fig.~3(a) shows the average LLM \ac{HO} delay by changing backhaul rate \(R_{\mathrm{bh}}\). \ac{tHO} remains unchanged since it does not rely on backhaul KV cache transfer. By contrast, the delay of \ac{cHO} decreases as \(R_{\mathrm{bh}}\) increases and approaches that of \ac{ctHO} in the high backhaul regime, whereas at \(R_{\mathrm{bh}}=2\)Gbps, \ac{cHO} exhibits more than \(3.1\times\) larger delay than \ac{ctHO}.


Fig.~3(b) varies the compute capability by changing $a$ in the prefill delay model $p(L)=aL+b$. \ac{cHO} is unchanged because it does not use prefill. \ac{tHO} improves as $a$ decreases since faster prefill reduces prefill delay. \ac{ctHO} consistently achieves the smallest delay by selecting $L^{\star}$ according to the compute condition: when prefill is slow (large $a$), it reduces $L^{\star}$ and relies more on backhaul transmission, and when prefill becomes faster (small $a$), it increases $L^{\star}$ to reduce the backhaul load.



In Fig.~3(c), the worst-user LLM \ac{HO} delay increases with the maximum cache size $C_{\max}$, since a larger context requires more computation and transmission before streaming can resume. Across the entire sweep, \ac{ctHO} consistently achieves the lowest delay by jointly leveraging prefill and KV cache transmission.

Fig.~3(d) shows the worst-user LLM HO delay as \(K\) increases. As \(K\) increases, the delay of all methods increases because more UEs must share the same backhaul and prefill resources. Among the three schemes, \ac{ctHO} consistently achieves the smallest delay over the entire range of \(K\). The performance gap becomes more pronounced as \(K\) increases. For example, at \(K=12\), the LLM \ac{HO} delay of \ac{ctHO} is about \(0.95\) s, compared with \(1.25\) s for \ac{tHO} and \(3.35\) s for \ac{cHO}.

\subsection{Comparison of Total Streaming Delay With and Without HO}
\label{sec:sim_wireless_noho}
\begin{figure}[t]
\centering
\includegraphics[width=0.5\columnwidth]{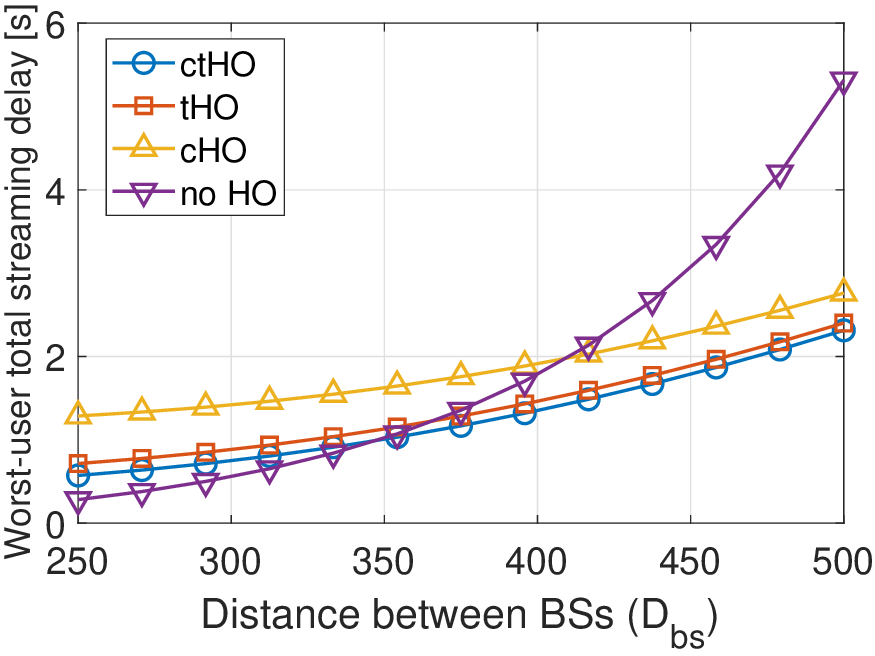}
\caption{Worst-user total streaming delay vs. distance between BSs with $R_{\mathrm{bh}}=4.5$\,Gbps. The total streaming delay includes both the LLM \ac{HO} delay and the subsequent streaming delay to deliver generated $G$ tokens after LLM \ac{HO}.}

\label{fig:plot_G_sweep_v1}
\end{figure}

To compare the delay with and without \ac{HO}, we evaluate the worst-user total streaming delay which includes LLM \ac{HO} delay and token streaming delay after LLM \ac{HO}, by sweeping $D_{\mathrm{bs}}$, while fixing the number of newly generated tokens after LLM \ac{HO} to \(G=1024\) tokens.
The \ac{SNR} of UE \(i\) at time \(t\) is modeled as a function of the distance between the UE position \(x_i(t)\) and the serving \ac{BS} location \(x^{(\mathrm{b})}\), where \(\mathrm{b}=\mathrm{s}\) with \(x^{(\mathrm{s})}=0\) denotes the no-\ac{HO} case and \(\mathrm{b}=\mathrm{t}\) with \(x^{(\mathrm{t})}=D_{\mathrm{bs}}\) denotes the \ac{HO} case:
\begin{align} \label{eq:snr}
\gamma_i^{(\mathrm{b})}(t)
= \gamma_{\mathrm{ref}}\Big(\frac{d_{\mathrm{ref}}}{|x_i(t)-x^{(\mathrm{b})}|}\Big)^{\alpha},
\end{align}
where \((\gamma_{\mathrm{ref}},d_{\mathrm{ref}},\alpha)=(10~\mathrm{dB},20~\mathrm{m},3.5)\). The corresponding wireless token streaming rate for each token is given by
\begin{align}
r^{(\mathrm{b})}\!\big(t\big)
\triangleq \frac{W \log_2\!\big(1+\gamma_i^{(\mathrm{b})}(t)\big)}{s_{\mathrm{tok}}} [\text{tokens/s}],
\end{align}
where a token size ${s_{\mathrm{tok}}} = $12 bits/token and \(W=2\) MHz.

The total streaming delay measured from \(\tau_i\) is defined as
\begin{align}
D_i(L,r,G)
\triangleq \big(T_i(L,r)-\tau_i\big) + D_i^{(\mathrm{st})}\!\big(G,r^{(\mathrm{b})}\!(t)\big) [\text{s}],
\label{eq:Di_total_ho}
\end{align}
where \(D_i^{(\mathrm{st})}(G;\gamma)\) denotes the delay required to stream newly generated tokens after LLM \ac{HO}. In the no-\ac{HO} case, the LLM \ac{HO} delay vanishes, and the total delay reduces to \(D_i^{(\mathrm{st})}\!\big(G,r^{(\mathrm{b})}\!(t)\big)\). We report the worst-user total streaming delay \(D^{(\mathrm{tot})}(D_{\mathrm{bs}})\triangleq \max_i D_i(L,r,G)\) while sweeping \(D_{\mathrm{bs}}\).


As \(D_{\mathrm{bs}}\) increases, the no-\ac{HO} delay grows much faster than the \ac{HO}-based methods because degraded link increases streaming delay. In particular, a larger \(D_{\mathrm{bs}}\) leads to a more severe SNR drop on the link between the source \ac{BS} and UEs. 
Although no-\ac{HO} achieves a smaller delay when \(D_{\mathrm{bs}}<375\) m due to the absence of LLM \ac{HO} delay, as \(D_{\mathrm{bs}}\) increases, \ac{HO}-based methods yield lower delay.
At \(D_{\mathrm{bs}}=500\) m, the no-\ac{HO} delay exceeds that of \ac{ctHO} by more than \(2.1\) s. Overall, \ac{ctHO} consistently achieves the smallest worst-user delay among the \ac{HO}-based methods over the entire $D_{\mathrm{bs}}$ range.

\section{Conclusion}\label{sec:conclusion}
This paper investigated mobility-aware Edge LLM \ac{HO} with multi-UE token streaming, where the target BS must recover the KV cache. We proposed \ac{ctHO}, a joint method that optimizes the amount to prefill and KV cache backhaul rate, and developed a step-wise solution that selects the prefill length and backhaul scheduling to minimize the worst-user LLM \ac{HO} delay. Simulations show that \ac{ctHO} consistently outperforms other baselines under the simulation setting.
Future work includes extending the current hard LLM \ac{HO} framework to a soft LLM \ac{HO} setting, where the target BS performs pre-computation and the source BS continues decoding for more seamless service continuity.

\bibliographystyle{IEEEtran}
\bibliography{ref}

@inproceedings{liu2024,
  title={Cachegen: Kv cache compression and streaming for fast large language model serving},
  author={Liu, Yuhan and Li, Hanchen and Cheng, Yihua and Ray, Siddhant and Huang, Yuyang and Zhang, Qizheng and Du, Kuntai and Yao, Jiayi and Lu, Shan and Ananthanarayanan, Ganesh and others},
  booktitle={Proceedings of the ACM SIGCOMM 2024 Conference},
  pages={38--56},
  year={2024}
}

@inproceedings{Jang25,
  title={Edge-first language model inference: Models, metrics, and tradeoffs},
  author={Jang, SiYoung and Morabito, Roberto},
  booktitle={2025 IEEE 45th International Conference on Distributed Computing Systems Workshops (ICDCSW)},
  pages={309--314},
  year={2025},
}

@misc{Yet2026,
      title={{SLA}-Aware Distributed LLM Inference Across Device-RAN-Cloud}, 
      author={Hariz Yet and Nguyen Thanh Tam and Mao V. Ngo and Lim Yi Shen and Lin Wei and Jihong Park and Binbin Chen and Tony Q. S. Quek},
      year={2026},
      eprint={2602.23722},
      archivePrefix={arXiv},
      primaryClass={cs.NI},
      url={https://arxiv.org/abs/2602.23722}, 
}

@techreport{WG3,
title={{AI}-on-{RAN}: Enabling Monetizable Differentiated Connectivity for AI},
author= {{AI-RAN Alliance}},
journal={Tech. Rep.},
year={2026},
url={https://ai-ran.org/documents/AI-RAN-WG3-AI-on-RAN-Whitepaper.pdf}, 
}

@misc{Mahdi,
      title={Token-Domain Multiple Access: Exploiting Semantic Orthogonality for Collision Mitigation}, 
      author={Li Qiao and Mahdi Boloursaz Mashhadi and Zhen Gao and Deniz Gündüz},
      year={2025},
      eprint={2502.06118},
      archivePrefix={arXiv},
      primaryClass={cs.IT},
      url={https://arxiv.org/abs/2502.06118}, 
}

@article{Lee2025,
  title={Low-complexity semantic packet aggregation for token communication via lookahead search},
  author={Lee, Seunghun and Park, Jihong and Choi, Jinho and Park, Hyuncheol},
  journal={arXiv preprint arXiv:2506.19451},
  year={2025}
}

@article{C2C,
  title={Cache-to-cache: Direct semantic communication between large language models},
  author={Fu, Tianyu and Min, Zihan and Zhang, Hanling and Yan, Jichao and Dai, Guohao and Ouyang, Wanli and Wang, Yu},
  journal={arXiv preprint arXiv:2510.03215},
  year={2025}
}

@inproceedings{Chen2026,
  title={Federated Inference for Heterogeneous {LLM} Communication and Collaboration},
  author={Chen, Zihan and Li, Zeshen and Yang, Howard Hao and Quek, Tony and Park, Jihong},
  booktitle={AAAI 2026 Workshop on Machine Learning for Wireless Communication and Networks (ML4Wireless)},
  year={2026}
}

@article{Qiao2025b,
  title={Token communications: A large model-driven framework for cross-modal context-aware semantic communications},
  author={Qiao, Li and Mashhadi, Mahdi Boloursaz and Gao, Zhen and Tafazolli, Rahim and Bennis, Mehdi and Niyato, Dusit},
  journal={IEEE Wireless Communications},
  volume={32},
  number={5},
  pages={80--88},
  year={2025},
  publisher={IEEE}
}

@article{Zheng2024,
  title={Batch{LLM}: Optimizing large batched {LLM} inference with global prefix sharing and throughput-oriented token batching},
  author={Zheng, Zhen and Ji, Xin and Fang, Taosong and Zhou, Fanghao and Liu, Chuanjie and Peng, Gang},
  journal={arXiv preprint arXiv:2412.03594},
  year={2024}
}

@article{Neetu2023,
  title={Performance analysis of cache-enabled handover management for vehicular networks},
  author={Neetu, RR and Ghatak, Gourab and Bohara, Vivek Ashok and Srivastava, Anand},
  journal={IEEE Transactions on Network Science and Engineering},
  volume={11},
  number={1},
  pages={1151--1164},
  year={2023},
  month={Oct.},
  publisher={IEEE}
}
\end{document}